\documentclass[envcountsame]{llncs}
\usepackage[utf8]{inputenc}
\usepackage{microtype}
\usepackage{amsmath,amssymb}
\usepackage{multirow}
\usepackage{graphicx}
\usepackage{array}
\usepackage{float}
\usepackage{todonotes}
\usepackage{wrapfig}
\usepackage{nicefrac}

\title{Approximating Node-Weighted $k$-MST\\ on Planar Graphs}
\author{%
  Jarosław Byrka\inst{1}  \and 
  Mateusz Lewandowski\inst{2} \and
  Joachim Spoerhase\inst{3}
}%
\institute{
  Institute of Computer Science, University of Wrocław, Poland,\\
  \email{jby@cs.uni.wroc.pl}
  \and
  Institute of Computer Science, University of Wrocław, Poland,\\
  \email{mlewandowski@cs.uni.wroc.pl}
  \and
  Lehrstuhl für Informatik I, Universität Würzburg, Germany,\\
  \email{joachim.spoerhase@uni-wuerzburg.de}
}

\spnewtheorem{fact}[theorem]{Fact}{\bfseries}{\itshape}
\spnewtheorem{observation}[theorem]{Observation}{\bfseries}{\itshape}

\newcommand{\ie}{\textit{i.e.}}
\newcommand{\eg}{\textit{e.g.}}
\newcommand{\eps}{\varepsilon}
\newcommand{\oh}{\mathcal{O}}
\newcommand{\opt}{\mathrm{OPT}}

\begin{document}

\maketitle

\begin{abstract}
  We study the problem of finding a minimum weight connected subgraph
  spanning at least $k$ vertices on planar, node-weighted graphs. We
  give a $(4+\eps)$-approximation algorithm for this problem. We
  achieve this by utilizing the recent LMP primal-dual
  $3$-approximation for the node-weighted prize-collecting Steiner
  tree problem by Byrka et al (SWAT'16) and adopting an approach by
  Chudak et al. (Math.\ Prog.\ '04) regarding Lagrangian relaxation
  for the edge-weighted variant.  In particular, we improve the
  procedure of picking additional vertices (tree merging procedure)
  given by Sadeghian (2013) by taking a constant number of recursive
  steps and utilizing the limited guessing procedure of Arora and
  Karakostas (Math.\ Prog.\ '06).

  More generally, our approach readily gives a
  $(\nicefrac{4}{3}\cdot r+\eps)$-approximation on any graph class
  where the algorithm of Byrka et al.\ for the prize-collecting
  version gives an $r$-approximation. We argue that this can be
  interpreted as a generalization of an analogous result by Könemann
  et al. (Algorithmica~'11) for partial cover problems.  Together with
  a lower bound construction by Mestre (STACS'08) for partial cover
  this implies that our bound is essentially best possible among
  algorithms that utilize an LMP algorithm for the Lagrangian
  relaxation as a black box. In addition to that, we argue by a more
  involved lower bound construction that even using the LMP algorithm
  by Byrka et al.\ in a \emph{non-black-box} fashion could not beat
  the factor $\nicefrac{4}{3}\cdot r$ when the tree merging step
  relies only on the solutions output by the LMP algorithm.
\end{abstract}

\section{Introduction}
  We consider the node-weighted variant of the well-studied
  $k$-MST problem. Given a graph $G=(V,E)$ with non-negative node
  weights $w\colon V\rightarrow \mathbb{R}_+$ and a positive integer
  $k$, we consider the problem of finding a minimum cost connected
  subgraph of $G$ spanning $k$ vertices. In analogy to the edge-weighted
  case, we call this problem node-weighted $k$-MST (NW-$k$-MST) because
  the solution can be assumed to be a tree. In fact, we focus on the
  rooted variant in which a given vertex $r$ has to be included
  in the final solution. To obtain the unrooted version, simply use
  resulting algorithm for each choice of root vertex.

  It was already observed that this problem
  is $\Omega(\log n)$-hard to approximate
  \cite{DBLP:journals/siamcomp/MossR07}. However, the problem becomes
  easier, when we restrict $G$ to be a planar graph. The focus of this work is
  to provide an approximation algorithm with small factor for this case.

  \subsection{Related work}
    \subsubsection{Edge-weighted $k$-MST}
    The standard, edge-weighted k-MST problem was thoroughly studied.
    In the sequence of papers
    \cite{DBLP:conf/focs/Garg96,DBLP:journals/mp/AroraK06,DBLP:conf/stoc/Garg05}
    the $2$-approximation algorithm for prize-collecting Steiner tree
    problem \cite{DBLP:journals/siamcomp/GoemansW95} was used to
    finally obtain a $2$-approximation algorithm for k-MST. These
    results can be, to some extent, explained as in Chudak et. al
    \cite{DBLP:journals/mp/ChudakRW04} in terms of Lagrangian
    Relaxation. 
    
    In particular the $5$-approximation algorithm follows the framework known mostly from
    Jain and Vazirani's work on the $k$-median problem
    \cite{DBLP:journals/jacm/JainV01}. In these algorithms, the
    Lagrangian multiplier preserving (LMP) property plays a crucial
    role. The LMP property is also satisfied by the Goemans-Williamson
    algorithm for the prize-collecting Steiner tree problem (PC-ST).
    Intuitively, the LMP property of an $\alpha$-approximation algorithm
    for some prize-collecting problem, means that the solutions it produces
    would also be not more expensive than $\alpha$ times
    optimum value even if we would have to pay $\alpha$ times more for penalties.
    \subsubsection{Node-weighted $k$-MST}
      The NW-$k$-MST problem was already studied in the more general quota
      setting, where each node has also associated some profit, and the goal
      is to find the minimum cost connected set of vertices having at least
      some total profit $\Pi$. In particular, $O(\log n)$-approximation
      was given in \cite{DBLP:journals/siamcomp/MossR07}.
      However, this result was based on their invalid
      $O(\log n)$-approximation for NW-PC-ST.
      Recently, Chekuri et. al \cite{DBLP:conf/approx/ChekuriEV12} and also
      independently Bateni et. al \cite{DBLP:conf/icalp/BateniHL13} proposed
      correct algorithms for generalizations of NW-PC-ST,
      but without LMP guarantee.
      The result on the quota problem was finally restored by
      K{\"{o}}nemann et al. \cite{DBLP:conf/focs/KonemannSS13} where an
      LMP algorithm was developed. In the related master thesis
      \cite{SinaMasterThesis}, Sadeghian gives also an alternative way of
      picking vertices\footnote{by picking vertices we mean augmenting
      the smaller solution with some vertices of larger solution.
      This is an important ingredient for the Lagrangian Relaxation technique}
      in the reduction for the quota problem. In these results,
      the constant lost in the process was not optimized.
    \subsubsection{Node-weighted planar Steiner problems} 
      Recently, the planar variants of Steiner problems received increased
      attention. In particular, Demaine et. al 
      \cite{DBLP:journals/talg/DemaineHK14} obtained a $6$-approximation
      for the node-weighted Steiner forest problem. The factor was further
      improved to $3$ by Moldenhauer~\cite{DBLP:journals/iandc/Moldenhauer13}.
      Both results rely on the moat-growing algorithm similar to that of
      Goemans and Williamson.
      Currently the best result for this problem is the $2.4$ approximation
      by Berman and Yaroslavtsev~\cite{BermanY12:nw-nd} where
      they use a different oracle for determining violated sets.
      
      More general network design problems on planar graphs where also studied
      by Chekuri et. al \cite{DBLP:conf/icalp/ChekuriEV12}. Finally,
      the result of Moldenhauer was generalized to the prize-collecting
      variant by Byrka et. al~\cite{DBLP:conf/swat/ByrkaLM16}, resulting in
      an LMP $3$-approximation for NW-PC-ST on planar graphs.
      We note that our result highly relies on this last algorithm.

      \subsubsection{Partial cover}
      Below, we argue that our problem on arbitrary graphs generalizes
      the \emph{partial cover} problem. In this problem we are given a
      set cover instance along with a positive integer $k$. The
      objective is to cover at least $k$ ground elements by a family
      of sets of minimum cost. In the \emph{prize-collecting} version
      of the problem every element has a penalty and the objective is
      to minimize the sum of costs of the chosen sets and the
      penalties of the elements that are not covered. Könemann et
      al.~\cite{KonemannPS11:Partial-Covering} describe a unified
      framework for partial cover. They show how to obtain an
      approximation algorithm for a class~$\mathcal{I}$ of partial
      cover instances if there is an $r$-approximate LMP algorithm for
      the corresponding prize-collecting version. In particular, their
      result implies a $(\frac{4}{3}+\eps)r$-approximation algorithm for the
      class~$\mathcal{I}$. Mestre~\cite{Mestre08:lagrangian-relax-partial-cover}
      shows that no algorithm that uses an LMP algorithm as a black
      box can obtain a ratio better than $\frac{4}{3}r$ so these results are
      essentially optimal.
      
  \subsection{Our result and techniques}
    We give a polynomial-time $(4+\eps)$-approximation algorithm for
    the NW-$k$-MST problem on planar graphs. Our result extends to
    an algorithm for the quota node-weighted Steiner tree problem on
    planar graphs with the same factor.

    The main technique we use is the Lagrangian relaxation framework
    (as mentioned in related works above) where two solutions --- one
    with less and the other with more than $k$ nodes --- are combined
    to obtain a feasible tree. The overview of our algorithm is as follows:
    \begin{enumerate}
      \item{guess a skeleton and prune the instance}
      \item{using the LMP algorithm~\cite{DBLP:conf/swat/ByrkaLM16}, find trees $T_1, T_2$ with
          $\leq k$ and $\geq k$ nodes, respectively}
      \item{combine $T_1$ and $T_2$ to a single tree with exactly $k$ vertices}
    \end{enumerate}
    This is the standard design (although guessing step is not always necessary)
    of algorithms based on Lagrangian relaxation framework.
    However, in order to optimize the constant we employ
    additional ideas and techniques.
    
    The first guessing step bears some similarities to that of Arora and
    Karakostas \cite{DBLP:journals/mp/AroraK06} where they improve
    Garg's $3$-approximation for edge-weighted $k$-MST to $2+\eps$. This
    additional guessing allows them to pay $\eps \cdot \opt$ instead
    of $\opt$ for connecting a single set of vertices to the rest of
    the solution. Here, we provide a node-weighted variant of this idea and
    also use it more extensively, because we have to buy multiple (but
    still a constant number of) such connections. In our approach,
    we guess a set of vertices from optimum solution and call it
    a skeleton. Then, we can safely prune the instance ensuring
    that each remaining node will be not too far away from the skeleton.
    The guessing step is described in Section~\ref{s:pruning}.
  
    For the second step, we have to slightly modify the primal-dual
    LMP $3$-approximation algorithm \cite{DBLP:conf/swat/ByrkaLM16},
    so it returns solutions containing the guessed skeleton. This
    modification is technical and is described --- together with the
    way of finding suitable $T_1$ and $T_2$ ---
    in Section~\ref{s:lagrangian-relaxation-moat-growing}.

    In the third step, we combine $T_1$ and $T_2$ by extending the procedure
    of picking vertices of Sadeghian~\cite{SinaMasterThesis}.
    He finds some cost-effective subset of vertices, which is two times
    larger than needed. We show that by picking vertices in certain order
    and applying recursion a \emph{constant} number of times, we are able
    to pick almost exactly the number of nodes that is needed. Although,
    the number of components of this set might be arbitrary, we need
    to buy only a constant number of connections to restore connectivity.
    This is our main contribution and is described
    in the Section~\ref{s:4-eps-approximation}.

    The resulting approximation factor of our algorithm is $(4+\eps)$.
    Additionally, we show some evidence that our combining step is
    in some sense optimal. More precisely, we show that no other algorithm,
    using LMP $3$-approximation as a black-box and not referring
    to the planarity can give better constant than $4$. This is obtained
    by interpreting our algorithm in terms of the results for the partial cover problem.
    The optimality of our algorithm within this framework is discussed in
    Section~\ref{s:partial-cover}.

\section{Pruning the Instance}\label{s:pruning}
First, we assume that we know $\opt$ up to a factor $1+\eps$ by
using standard guessing techniques~\cite{DBLP:conf/focs/Garg96}. A node $v$ is called \emph{$\eps$-distant} to a node set $U\subseteq V$ if there exists a path $P$ in $G$ from $v$ to a node $u\in U$ of node weight $c(V(P)\setminus\{u\})\leq\eps\cdot\opt$.

\begin{lemma}\label{lem:guessing-an-eps}
  Consider an optimum solution $T$ and an $\eps>0$. Then there exists
  a set $W\subseteq V(T)$ of size at most $1/\eps$ such that each node
  in $T$ is $\eps$-distant to $W\cup\{r\}$.
\end{lemma}
\begin{proof}
  Consider $T$ as a tree rooted at $r$.  For any node $u$ in this tree
  let $T_u$ denote the subtree hanging from $u$. A subtree $T_u$ is
  called \emph{good} if for any node in $T_u$ the node weight
  (including the weight of the end nodes) of the unique path from this
  node to $u$ within $T_u$ is at most $\eps\cdot\opt$.

  We traverse $T$ in a bottom-up fashion starting with the leaves. We
  maintain the invariant (by removing subtrees) that for all nodes $u$
  visited so far and still being in $T$, the subtree $T_u$ is good. To
  this end, when we encounter a node $u$ such that $T_u$ is good we
  just continue with the traversal. If $T_u$ is bad, however, then
  there must be a path $P$ within $T_u$ ending in $u$ of node weight
  $c(P)\geq \eps\cdot\opt$. We include $u$ into $W$ and assign $P$ as
  a \emph{witness} to $u$. Because of our invariant for all (if any)
  children $v$ of $u$, we have that $T_v$ is good. This means in
  particular that for all nodes $z$ in $T_u$ the node weight
  (\emph{excluding} the weight of $u$) of the path from $z$ to $u$ is
  at most $\eps\cdot \opt$. Finally, remove $T_u$ from $T$ and
  continue with the traversal. We stop when we reach the root $r$ at
  which point we remove the remaining tree (for the sake of analysis).

  First, note that the set $W$ has cardinality at most $1/\eps$
  because we assigned to each node in $W$ a witness path of weight at
  least $\eps\cdot\opt$ and because the witness paths are pairwise
  node-disjoint. Second, observe that whenever we removed a node $z$
  from $T$ as part of a subtree $T_u$, the node weight (excluding the
  weight of $u$) of the path from $z$ to $u$ was at most
  $\eps\cdot\opt$. Hence, for every node in $T$ there exists such a
  path to a node in $W\cup\{r\}$ at the end of the tree traversal
  since every node was removed.\qed
\end{proof}
In the sequel, we will call such a set $W$ whose existence is provided
by the above lemma an \emph{$\eps$-skeleton}.

In a pre-processing, we iterate over all $n^{\oh(1/\eps)}$ many
sets $W'\subseteq V$ with $|W'|\leq 1/\eps$ thereby guessing the
$\eps$-skeleton $W$ whose existence is guaranteed by the above
lemma. Moreover, we prune all nodes $u$ from the instance that are not
$\eps$-distant to $W\cup\{r\}$.

\section{The $(4+\eps)$-Approximation Algorithm}
  \label{s:4-eps-approximation}
  In Chapter 3 of \cite{SinaMasterThesis} Sadeghian describes a $O(\log n)$
  approximation for node-weighted quota Steiner tree problem. His result is
  established using a framework of \cite{DBLP:journals/mp/ChudakRW04},
  repeated also in \cite{DBLP:journals/siamcomp/MossR07} where a primal-dual
  LMP approximation algorithm for the prize-collecting Steiner tree problem can be
  used along with the Lagrangian relaxation method to obtain an approximation
  algorithm for the quota version of the problem. Sadeghian looses some large
  constant factor in the process. Direct application of his result would
  yield two digit approximation factor for our problem.
  
  We now show that carefully injecting the LMP 3-approximation algorithm for
  NW-PC-ST on planar graphs given in \cite{DBLP:conf/swat/ByrkaLM16} into his
  analysis yields a $(4+\eps)$-approximation. However, in the process, we need
  a more efficient way to pick additional vertices. We show that it is possible
  to pick a cheap set of these vertices. Although it will not be connected,
  only a \begin{it}constant\end{it} number of additional $\eps$-distant vertices will suffice to restore the connected tree.

  For ease of the presentation, we will focus on the NW-$k$-MST problem.
  The algorithm for quota version can be then easily deduced by arguments
  of Bateni et. al~\cite{DBLP:conf/icalp/BateniHL13}

  The analysis relies on the following lemma.
  \begin{lemma}
  \label{lem:framework}
    We can produce trees $T_1$ and $T_2$ containing all the vertices $W$ from
    the $\eps$-skeleton and the root $r$ of sizes $|T_1|\leq k \leq |T_2|$, such
    that for $\alpha_1,\alpha_2\geq 0$ with $\alpha_1+\alpha_2=1$ and $\alpha_1|T_1| + \alpha_2|T_2| = k$ we have that
    $$\alpha_1 c(T_1) + \alpha_2 c(T_2) \leq (3+\eps) OPT$$
  \end{lemma}
  The construction of these trees $T_1$ and $T_2$ and the proof of above lemma
  is described in Section~\ref{s:lagrangian-relaxation-moat-growing}.

  Let now $q=k-|T_1|$ be the number of vertices that are missing from the
  tree $T_1$. We will now show, that these vertices can be picked from
  $T_2\setminus T_1$ without paying too much.
  \begin{lemma}\label{l:merging}
    It is possible to find a (not necessarily connected) set $S$
    of at least $q$ vertices in $T_2\setminus T_1$ of cost at most
    $(1+\eps_2) \alpha_2 c(T_2)$,
    which can be connected to $T_1$ by connecting additionally $\oh(\log(1/\eps_2))$ many
    $\eps$-distant vertices  to the $\eps$-skeleton, where $\eps_2$ is any constant.
  \end{lemma}
  \begin{proof}
    Here, we substantially extend the analysis in \cite{SinaMasterThesis}.
    Consider a graph $T'_2$ constructed from $T_2$ by contracting all
    vertices from $T_1 \cap T_2$ to a single vertex $r'$. Define the cost
    of this vertex $r'$ to $0$ (we will buy $T_1$ anyway). From now on,
    whenever we count the cardinality of some subset $S$ of vertices in $T'_2$,
    we do not count vertex $r'$.
    \begin{definition}
    A subset of vertices $S$ is cost-effective if
    $\frac{c(S)}{|S|} \leq \frac{c(T'_2)}{|T'_2|}$.
    \end{definition}
    \begin{lemma}
    If cost-effective set $S$ has size $(1+\eps_2)q$ then its cost
    is at most $(1+\eps_2) \alpha_2 c(T_2)$.
    \end{lemma}
    \begin{proof}
    $$
      c(S) \leq |S| \frac{c(T'_2)}{|T'_2|}
           \leq (1+\eps_2) q \frac{c(T_2)}{|T_2|-|T_1|}
           \leq (1+\eps_2) \alpha_2 c(T_2),
    $$
    where we used the fact that $\alpha_2=\frac{k-|T_1|}{|T_2|-|T_1|}$.\qed
    \end{proof}
    So now, our goal is to find a cost-effective set $S$ in $T'_2$ of size
    only slightly larger that $q$. First, we start with a procedure
    for picking at most $2q$ vertices as in \cite{SinaMasterThesis}.
    Initialize graph $H$ with any spanning tree of $T'_2$. Observe that
    $H$ is cost-effective. Consider any edge $e$ of $H$.
    Let $X$ and $Y$ be the two components that would be created
    after removing the edge $e$ from $H$. At least one of these two components
    must be cost-effective. For any cost-effective component from this two,
    say $X$, do the following. If $X$ has enough vertices, $\ie$ $|X|\geq q$,
    remove $Y$ from $H$ and continue. Otherwise, contract vertices of $X$
    to a single super vertex and set its cost to the sum of all vertices
    in $X$. Also, define the super-cardinality of this new super vertex
    to $|X|$.
    
    It can be seen that after repeating this procedure as many times
    as possible, the graph $H$ will be a star graph with super-cardinality of each leaf
    at most $q$. Let $p$ be the number of leaves of $H$.
    In the case when $p \leq 1$ it is easy to see, that taking the
    whole graph $H$ would result in a cost-effective set of vertices
    of size at most $2q$. Therefore, assume now that $p \geq 2$.
    Then, there exists a central vertex of the star graph $H$,
    call it $c$, which is not a super vertex.
    Moreover, every leaf $v$ must be cost-effective
    (otherwise either we would remove $v$, or $H$ would consist of two
    nodes). Observe also, that the super-cardinality of each leaf
    is at most $q$. Hence adding leaves to $S$ one by one, would eventually
    lead to the set $S$ with super-cardinality at most $2q$ (and at least $q$).
    Finally, $S$ could be connected to $T_1$ by a single path from vertex $c$.

    We now modify this procedure of adding leaves. First, consider them in the
    order of decreasing super-cardinalities. To this end, 
    let $v_1, v_2, \dots v_p$ be leaves of $H$ and 
    $s_1 \geq s_2 \geq \dots \geq s_p$ be the corresponding super-cardinalities.
    Find the smallest $i$ such that
    $\sum_{j=1}^{i} s_j + s_{i+1} \geq q$.
    If $s_{i+1} = 1$, then the desired set $S$ consist of all vertices in
    $v_1, v_2, \dots v_{i+1}$ and it has exactly $q$ vertices.
    Otherwise, add the first $i$ leaves to the set $S$.
    Let $t=\sum_{j=1}^{i} s_j$ be the number of vertices added to $S$.
    Now, instead of adding to $S$ all vertices in the super vertex $s_{i+1}$,
    we expand this super vertex back to the original graph and repeat the above
    process with the new number of vertices to pick equal to $q' = q-t$.
    Observe that, because of sorting we have that $t \geq \frac{1}{2}q$,
    which also implies that $q' \leq \frac{1}{2}q$.
    This process is repeated recursively up to $l$ times---where $l$ is a parameter---
    but in the last call we take the last leaf completely.

    Let now $q_1,q_2,\dots,q_l$ be the numbers of vertices to pick in
    respective recursive calls
    (note that $q_1=q$ and $q_j \leq \frac{1}{2} q_{j-1}$).
    The total number of picked vertices is then at most
    $q+2q_l \leq (1+2^{-l+2})q$. Therefore, to find
    the desired set $S$ of at most $(1+\eps_2)q$ vertices, we need only
    a constant number of recursive calls ---
    parameter $l$ is only $\oh(\log(1/\eps_2))$.
    Moreover all the vertices of $S$
    can be connected to $T_1$ by buying paths from the central nodes of
    all the $l$ star graphs that appeared in the process.
    This finishes the proof.\qed
  \end{proof}
  
  To construct a feasible solution, take the set $S$ guaranteed by the
  above lemma and connect it to $T_1$ by the $\oh(\log(1/\eps_2))$
  shortest paths to the $\eps$-skeleton. Denote this solution by
  $\mathrm{SOL}_1$.  Let also $\mathrm{SOL}_2$ be the entire tree
  $T_2$.  Our algorithm outputs cheaper of the two solutions
  $\mathrm{SOL}_1$ and $\mathrm{SOL}_2$.

This enables us to prove the following.
   \begin{lemma}\label{lem:four+eps-appr-algor}
     Assuming $\eps\leq 1$, the cost of the cheaper of the two
     solutions $\mathrm{SOL}_1$ and $\mathrm{SOL}_2$ is
     $(4+O(\sqrt{\eps}))\cdot\opt$.
   \end{lemma}
   \begin{proof}[Sketch]
     The derivations are similar to those of Könemann et al. and we thus refer for the full version to Appendix~\ref{sec:proof-lemma-refl}. Using $\alpha=\alpha_2$ and $\beta=c(T_1)/\opt$ we can infer
     \begin{align*}
     c(\mathrm{SOL}_1) & \leq \left(3(1+\eps_2) + \alpha \beta \right)
     \cdot \opt + \eps \cdot \oh(\log(1/\eps_2)) \cdot \opt\textnormal{,\quad and}\\
     c(\mathrm{SOL}_2) & = c(T_2) \leq
     \frac{3(1+\eps)-(1-\alpha)\beta}{\alpha} \cdot \opt\,.
   \end{align*}
   Balancing these various parameters
   $\alpha,\beta,\epsilon,\epsilon_2$ yields the bound.\qed
   \end{proof}

\section{Lagrangian Relaxation and Moat Growing on Planar Graphs}
  \label{s:lagrangian-relaxation-moat-growing}
  In this section we prove Lemma~\ref{lem:framework}. The proof utilizes
  Lagrangian Relaxation and follows a framework
  similar to the one in~\cite{DBLP:journals/mp/ChudakRW04}.

  We start with the following LP relaxation for the NW-$k$-MST problem, where solutions 
  are additionally constrained to contain all guessed vertices $W$
  of the $\eps$-skeleton. For each vertex $v$ we have the $x_v$ variable
  indicating whether we will include this vertex in the solution.
  The $z$ variables are indexed by sets of vertices not containing
  the root and the guessed vertices. In the optimum integral solution,
  only the one $z_X$ variable is set to $1$. This would be for the set
  $X$ of vertices not included in the final solution.

  \begin{align}
    \text{min } & \sum_{v \in V\setminus\{r\}} x_v c_v & (LP)      \nonumber\\
    s.t.                                                           \nonumber\\
      & \sum_{v \in \Gamma(S)} x_v
      + \sum_{\substack{X:S\subseteq X \\ X\cap W=\emptyset}}z_X \geq 1 &
        \forall S \subseteq V\setminus\{r\}                        \nonumber\\
      & x_v + \sum_{\substack{X:v\in X \\ X\cap W=\emptyset}}z_X \geq 1 &
        \forall v \in V\setminus\{r\}                              \nonumber\\
      & \sum_{X \subseteq V\setminus\{r\}} |X|z_X \leq n-k
                                                     \label{con:cardinality}\\
      & x_v \geq 0 &
        \forall v \in V\setminus\{r\}                              \nonumber\\
      & z_X \geq 0 &
        \forall X \subseteq V\setminus\{r\}                        \nonumber
  \end{align}
  The first two types of constraints guarantee connectivity of the solution to the root
  vertex and skeleton $W$. The $\Gamma(S)$ denotes the neighborhood of the set
  $S$, $\ie$ the set of vertices that are not in $S$, but have a neighboring
  vertex in $S$.

  The constraint~(\ref{con:cardinality}) ensures that the final solution will have
  at least $k$ vertices and introduces difficulties. Therefore, we move it
  to the objective function obtaining the following
  Lagrangian Relaxation:

  \begin{align*}
    \text{min } & \sum_{v \in V\setminus\{r\}} x_v c_v
    + \lambda\left(\sum_{X \subseteq V\setminus\{r\}} |X|z_X - (n-k)\right) &
      \left(LR(\lambda)\right)                                              \\
    s.t.                                                                    \\
      & \sum_{v \in \Gamma(S)} x_v
      + \sum_{\substack{X:S\subseteq X \\ X\cap W=\emptyset}}z_X \geq 1 &
        \forall S \subseteq V\setminus\{r\}                                 \\
      & x_v + \sum_{\substack{X:v\in X \\ X\cap W=\emptyset}}z_X \geq 1 &
        \forall v \in V\setminus\{r\}                                       \\
      & x_v \geq 0 &
        \forall v \in V\setminus\{r\}                                       \\
      & z_X \geq 0 &
        \forall X \subseteq V\setminus\{r\}
  \end{align*}
  The above LP (ignoring the constant $-\lambda(n-k)$ term in the objective
  function) is exactly the LP for the node-weighted prize-collecting Steiner
  tree (NW-PC-ST in which the penalty of each vertex in $V'=V \setminus W$ is equal to the
  parameter $\lambda$) with a slight modification that the subset of vertices
  $W$ is required to be in the solution.

  Consider now, the dual of the $LR(\lambda)$:
  \begin{align*}
  \text{max } & \sum_{S \subseteq V\setminus\{r\}} y_S
  + \sum_{v \in V\setminus\{r\}} p_v - \lambda(n-k) &
    \left(DLR(\lambda)\right)                                               \\
  s.t.                                                                      \\
    & \sum_{S: v \in \Gamma(S)} y_S
    + p_v \leq c_v &
      \forall v \in V\setminus\{r\}                                         \\
    & \sum_{X \subseteq S} y_X + \sum_{v \in S} p_v \leq \lambda |S| &
      \forall S \subseteq V'\setminus\{r\}                                  \\
    & y_S \geq 0 &
      \forall S \subseteq V\setminus\{r\}                                   \\
  \end{align*}
  Now, the slightly modified primal-dual LMP $3$-approximation for (NW-PC-ST)
  given in~\cite{DBLP:conf/swat/ByrkaLM16} can be used with penalties $\lambda$
  to produce the tree $T^\lambda$ and the dual solution $(y^\lambda,p^\lambda)$
  such that
  \begin{align}
    c(T^\lambda) + 3 \lambda(n-|T^\lambda|) \leq 3
    \left(\sum_{S \subseteq V\setminus\{r\}} y^\lambda_S
    + \sum_{v \in V\setminus\{r\}} p^\lambda_v\right) ,
    \label{e:lmp-guarantee}
  \end{align}
  where $T^\lambda$ contains all vertices of $W$.
  The description of this algorithm is deferred to
  Subsection~\ref{ss:moat-growing}. Let us now see how we
  can use it to finish the proof of Lemma~\ref{lem:framework}.
  We proceed essentially as in \cite{SinaMasterThesis} and
  \cite{DBLP:journals/mp/ChudakRW04}. By subtracting
  $3\lambda(n-k)$ from both sides of inequality~(\ref{e:lmp-guarantee})
  and simplifying the notation so that
  $\mathrm{DS}_\lambda = 
  \sum_{S \subseteq V\setminus\{r\}} y^\lambda_S
  + \sum_{v \in V\setminus\{r\}} p^\lambda_v$ denotes the value of a
  dual solution we have that
  \begin{align*}
    c(T^\lambda) + 3 \lambda(k-|T^\lambda|) &
      \leq 3 \left(\mathrm{DS}_\lambda - \lambda(n-k)\right) \\
    & \leq 3 \cdot \mathrm{DLR}(\lambda) \leq 3\cdot\opt .
  \end{align*}
  Observe that for $\lambda=0$ the algorithm could output a tree with at least
  $k$ vertices (because of moats growing around vertices in $W$, see next subsection).
  In this case the resulting tree is a $3$-approximation 
  so we do not need the merging procedure described
  in Section~\ref{s:4-eps-approximation}.
  Otherwise, for some large $\lambda$, $\eg$ the maximum cost of a vertex,
  the resulting tree would contain all the vertices.
  Therefore, we do the binary search for $\lambda$ such that
  $|T^\lambda|$ is close to $k$.
  In a lucky event $|T^\lambda|=k$ and then we don't need the merging
  procedure described in Section~\ref{s:4-eps-approximation}.
  Otherwise, we obtain $\lambda_1$ and $\lambda_2$ such that
  $|T^{\lambda_1}| < k < |T^{\lambda_2}|$. By making enough steps of
  the binary search we can ensure that
  $\lambda_2 - \lambda_1 \leq \frac{\eps \cdot \opt}{3n}$.
  Let these trees be $T_1$ and $T_2$.
  Now, by setting
  $\alpha_1=\frac{|T_2|-k}{|T_2|-|T_1|}$ and
  $\alpha_2=\frac{k-|T_1|}{|T_2|-|T_1|}$ and
  using inequality~(\ref{e:lmp-guarantee}) twice we have that
	\begin{align*}
		\alpha_1c(T_1) + \alpha_2c(T_2)&\leq
		  3\left(\alpha_1\mathrm{DS}_1 + \alpha_2\mathrm{DS}_2
      - \alpha_1\lambda_1(n-|T_1|) - \alpha_2\lambda_2(n-|T_2|)\right)\\
		&\leq
      3\left(\alpha_1\mathrm{DS}_1 + \alpha_2\mathrm{DS}_2
      - \lambda_2(n-k) + (\lambda_2-\lambda_1)(n-|T_1|)\right)\\
    &\leq
		  3\left(\opt+(\lambda_2-\lambda_1)n\right)\\
		&\leq
		  \left(3+\eps\right) \opt ,
	\end{align*}
  where we used the fact that the convex combination of $\mathrm{DS}_1$
  and $\mathrm{DS}_2$ is a feasible solution for $\mathrm{DLR}(\lambda_2)$.

  \subsection{Moat Growing} \label{ss:moat-growing}
    In this subsection we describe the slight technical modification needed
    in the primal-dual algorithm for NW-PC-ST problem on planar graphs given
    in \cite{DBLP:conf/swat/ByrkaLM16}. We also give a short description
    of the resulting algorithm for completeness.
    Observe, that there are two differences in the LPs used.

    First, we have additional constraints and
    corresponding dual variables $p_v$. This is due to the fact, that in our
    setting all vertices can have both nonzero penalty and cost, while in the
    previous setting the reduction step was employed so that each vertex is
    a terminal with some penalty and zero cost or a Steiner vertex with zero
    penalty. However, this reduction step is equivalent to setting $p_v$
    to minimum of cost and penalty and defining the reduced costs
    and reduced penalties. This does not influence the approximation factor, nor
    the LMP guarantee. See also Section~2.1 of Sadeghian~\cite{SinaMasterThesis}
    for details.

    The second modification comes from the fact that we have to include some guessed
    vertices $W$ in the solution. However, it is enough to treat these vertices
    in the same way as terminals.

    Now, we give a description of the algorithm.
    First, we do the mentioned reduction of eliminating $p_v$ variables. This
    makes some vertices terminal and the other Steiner vertices. We also
    add all the guessed vertices to the set of terminals and set their penalty
    to infinite.

    The algorithm maintains a set of moats, $\ie$, a family of disjoint sets
    of vertices. In each step, these moats can be viewed as the components
    of the graph induced by the so far bought nodes. Each moat has an associated potential
    equal to the total penalty of vertices inside this moat minus the sum of
    the dual variables for all the subsets of this moat. The moats with
    positive potential are active, with an exception that the moat containing the
    root is always inactive.

    The algorithm raises simultaneously the dual variables of all the active moats. For the growth
    of a moat we pay with its potential. We can have two events.
    
    In the first event, some vertex goes tight, $\ie$, the
    inequality for this vertex in the dual program becomes tight.
    In this case we buy this vertex and merge all the neighboring moats,
    setting the potential accordingly to the sum of all previous moats' potentials. We
    declare this new moat inactive whenever it contains a root vertex.

    In the second event, some moat goes tight, $\ie$ the inequality in the dual
    program becomes tight for some set of vertices. This corresponds to the
    situation when the potential of this moat drops to zero. In this case
    we declare this moat inactive and we mark all the previously unmarked
    terminals inside it as marked with the current time. Observe that
    in the dual we do not have these inequalities for sets containing guessed
    vertices $W$. This means, that all the vertices of $W$ will be connected
    to the root vertex.

    We repeat this process until we do not have any active moats. Then we
    start a pruning phase. We consider all the bought vertices in the reverse
    order of buying. We delete a vertex $v$ if the removal of $v$ would
    not disconnect any unmarked terminal or any terminal marked with time
    greater than the time of buying the vertex $v$. We return the pruned set
    of bought vertices as the solution.
    
    A straightforward adaptation of the analysis in \cite{DBLP:conf/swat/ByrkaLM16} implies that the above
    algorithm run with initial penalty $\lambda$ for all vertices in $V'$
    returns a tree $T^\lambda$ satisfying
    inequality~(\ref{e:lmp-guarantee}).

    \subsection{Generalization to non-planar graph classes}
   Note that in our algorithm, we use planarity exclusively by
   exploiting that the LMP algorithm of Byrka et
   al.~\cite{DBLP:conf/swat/ByrkaLM16} for the prize-collecting
   version has ratio~$3$ on planar graphs. Their algorithm, however,
   can be executed on an arbitrary graph class. Thus all our
   calculations can be carried through by replacing $3$ with any
   factor $r\geq 1$ thereby obtaining the following generalization.
    \begin{corollary}
      The above algorithm has performance $(\nicefrac{4}{3}+\eps)r$
      for any graph class where the algorithm of Byrka et
      al.~\cite{DBLP:conf/swat/ByrkaLM16} has a performance ratio of
      $r$.
    \end{corollary}

  \section{Trying to beat the factor of 4: relation to the partial cover} \label{s:partial-cover}
    Here we draw connections to the recent work on the partial cover
    problems.
    Könemann et al.~\cite{KonemannPS11:Partial-Covering} showed how to obtain
    a $(\nicefrac{4}{3}+\eps)r$-approximation algorithm for the partial cover
    problems using an $r$-approximate LMP algorithm for the corresponding
    prize-collecting version as a black-box. Their approach is roughly
    as follows. First, the most expensive sets from the optimum solution
    are guessed and all sets which are more expensive are discarded.
    Further, the black-box algorithm is used together with binary search
    to find two solutions, one, say $S_1$, feasible but possibly expensive,
    and the other, say $S_2$, infeasible but inexpensive.
    Then the merging procedure is employed to obtain a solution $S_3$.
    Finally, the cheapest solution of the $S_1$ and $S_3$ is returned.

    \subsection{Generalizing the Algorithm of Könemann et al.}
    Extending a folklore reduction from set cover type problems to
    node-weighted Steiner tree problems, we argue that our algorithm
    may be interpreted as a non-trivial generalization of the
    above-outlined algorithm by Könemann et
    al.~\cite{KonemannPS11:Partial-Covering}.

    First of all, the following reduction shows that the partial
    covering problem can be encoded as the quota node-weighted Steiner
    tree problem.  The reduction creates for each element a vertex
    with zero cost and profit $1$. Then, for each set it creates a
    node with the same cost and zero profit and connects it to the
    elements covered by this set.  Finally, the root vertex is added
    and connected to all the set-corresponding nodes. The target quota
    profit is set to be the same as the requirement for the partial
    cover problem.

    For such a reduced instance, we can run the preprocessing step
    from Section~\ref{s:pruning} which will remove the expensive sets
    (we could also employ the Könemann's preprocessing beforehand).
    Then, we would run any LMP algorithm for the prize-collecting
    cover problems within the Lagrangian relaxation framework which
    would indicate two families of sets to merge. Putting it on the
    reduced instance, these would correspond to two trees to merge.
    More precisely, take to the tree the set-corresponding nodes, the
    root vertex and the elements covered by sets.  Now, we can apply
    the merging procedure described in the Lemma~\ref{l:merging} with
    a slight adjustment needed to account for quota variant. In
    particular we modify the notion of cost-effectiveness to account
    profits instead of cardinalities and we also redefine the
    super-cardinality to be the sum of profits.  To retrieve the
    solution from the tree, simply take the sets corresponding to
    non-zero cost nodes in the tree. Finally, output the cheaper of
    the two feasible solutions giving a partial cover with the same
    quality as the one by obtained via the algorithm by Könemann et
    al.

    We remark that the above argument does not work in the reverse
    direction. The graph instances that are created have a very
    specific structure with three node layers ensuring that any
    partial cover solution is automatically connected at no additional
    cost. Achieving connectivity for \emph{general} graphs, however,
    is not implied and guaranteeing this structural property without
    loss in the performance guarantee of the algorithm can be seen as
    a main contribution of our work.

    \subsection{Black-box optimality}
    \label{s:black-box}
    \begin{wrapfigure}[11]{R}{0.4\textwidth}
    \centering
    \vspace{-30pt}
    \includegraphics[width=0.35\textwidth]{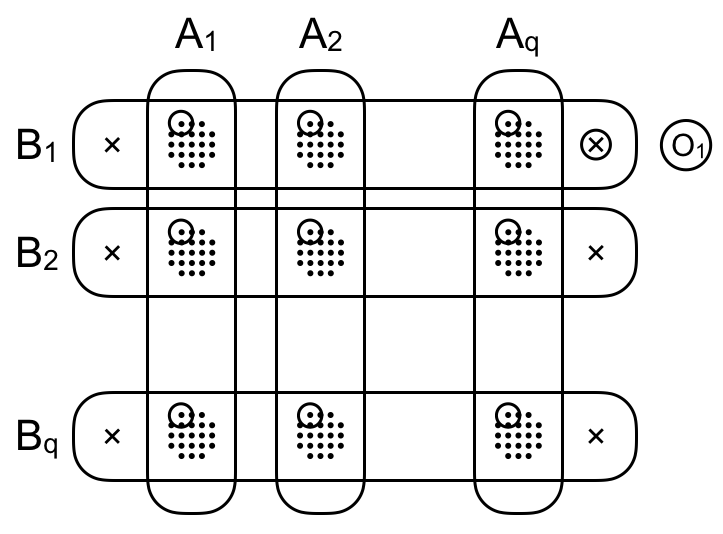}
    \caption{\label{fig:mestre-example}The instance of partial cover given
    by Mestre~\cite{Mestre08:lagrangian-relax-partial-cover}}
    \end{wrapfigure}
    Now, the above reduction, together with a lower bound construction
    by Mestre~\cite{Mestre08:lagrangian-relax-partial-cover} implies
    that our approach is best possible using the LMP algorithm as a
    black-box and without referring to the underlying graph class.  To
    see this, observe, that the Mestre's construction given in Theorem 3.1
    in~\cite{Mestre08:lagrangian-relax-partial-cover}, can be transformed
    to an instance of quota node-weighted Steiner tree instance by
    using the above reduction.

    Here, we repeat the Mestre's example, as we will extend it further.
    Fix some integer constant $q$.
    The instance consists of $q^3$ ground elements aligned
    in the grid of size $q$ by $q$ with $q$ elements in each cell.
    Then we have $q$ sets $A_1,A_2,\cdots A_q$, each covers all elements in the
    corresponding column of a grid. Analogously, we have $q$ sets
    $B_1, B_2, \cdots B_q$ which cover rows. Moreover, each $B_i$ set has
    two more ground elements. Then, we have $q$ sets $O_1, O_2, \cdots O_q$,
    where set $O_i$ covers $i$-th element from each cell of the grid and
    a single element which is also covered by $B_i$. This construction
    is illustrated in Figure~\ref{fig:mestre-example}, where the set $O_1$
    is marked with circles. Finally, costs of sets are defined as follows:
    $c(A_i)=\frac{2}{3}\cdot\frac{r}{q}$,
    $c(B_i)=\frac{4}{3}\cdot\frac{r}{q}$,
    $c(O_i)=\frac{1}{q}$, where $r=3$ in our case.

    Although, the corresponding quota Steiner instance is not planar
    and our algorithm does not exemplify the proof of his Lemma 3.3
    \footnote{This lemma states that there exists an LMP algorithm
      which returns either sets $A$ or $B$ (depending on the initial
      penalty $\lambda$).}, this example still shows that in order to
    beat the factor $\nicefrac{4}{3}\cdot r$ for NW-$k$-MST, we would
    indeed need to further consider the inner-workings of the LMP
    algorithm.  For details regarding this construction, we refer the
    original work of
    Mestre~\cite{Mestre08:lagrangian-relax-partial-cover}.

    \subsection{Inner-workings are not enough}
    Finally, we extend Mestre's example to show that even
    examining the inner-workings of the algorithm of
    Section~\ref{ss:moat-growing} without referring to the underlying
    graph class (such as planar graphs) in the merging procedure is
    not enough to beat the factor of $\nicefrac{4}{3}\cdot
    r$. We do this by giving a similar construction for which 
    Lemma 3.3 of \cite{Mestre08:lagrangian-relax-partial-cover} is
    satisfied, $\ie$ the LMP algorithm of
    Section~\ref{ss:moat-growing} returns either $A$ or $B$ sets. We
    will work with the instance of node-weighted prize-collecting
    Steiner tree problem obtained from Mestre's construction via our
    reduction. But first, we introduce two gadgets that are required
    for the final construction.

    \begin{subsubsection}{The potential aggregation}
      Recall, that the LMP algorithm grows moats around terminals until
      they run out of the initial potential. In the construction we will
      need two kinds of terminals. The first type, call it low-potential
      vertices are meant to become inactive very early.
      The second type of terminals, call them high-potential,
      are supposed to be active all the time
      during the run of the algorithm, $\ie$ until they connect
      to the root vertex.

      This differentiation can be easily achieved by connecting to a prospective
      high-potential vertex, a lot of new vertices.
      Then in the beginning of the GROW phase, the algorithm will make
      out of them a single component with large potential.
    \end{subsubsection}

    \begin{subsubsection}{The handicap gadget}
      We introduce a gadget which allows to significantly reduce the
      buying time of expensive $A$ and $B$ vertices so that they go tight
      at the same time and also much earlier than the cheaper
      $O$-vertices would.

      The gadget consist of a grid of vertices with $q$ columns and $q^2$
      rows. Each $B$-vertex is connected to every vertex of a grid.
      Each $A$-vertex is connected only to all vertices inside $\frac{q}{2}$
      columns. These columns are assigned in a way that each column is assigned
      to at least one $A$ vertex.
      Finally each vertex $O_i$ is connected to all vertices in column $i$.

      It can be seen that in the GROW phase of the LMP algorithm,
      the $B$ vertices gain their contribution to cost two times faster than
      $A$ vertices. Since $B$ vertices are twice as expensive,
      after adding this gadget, the buying time of $A$ and $B$
      should be now roughly the same.
    \end{subsubsection}

    \begin{subsubsection}{Finishing the construction}
      Here, we describe the final construction and analyze the behavior of
      the algorithm from Section~\ref{ss:moat-growing} on this instance.
      We extend the instance from Section~\ref{s:black-box}. Recall, that
      each set correspond now to a Steiner vertex which is also directly
      connected to the root vertex.

      Now, let the vertices which are in the $B$-sets and not in $A$-sets
      ($\ie$ these marked with cross in the Figure~\ref{fig:mestre-example})
      be the only low-potential vertices. Let the all
      other element-corresponding vertices be high-potential vertices.
      On top of that construction, add also the handicap gadget, in which
      each vertex of a grid is also a high-potential vertex.

      Now, the buying time of $A$ and $B$ vertices should be roughly the same.
      However, we insist that $A$ vertices should be bought first, hence we
      introduce some small perturbations to costs of $A$ vertices, $\ie$
      we subtract small $\eps$ from their costs.

      Set now target $k$ appropriately, $\ie$ $k=(2 \cdot q^3)\cdot \gamma + q$,
      where $\gamma$ is the number of additional vertices required for one
      high-potential vertex.

      Now, it can be seen, that there is an initial potential $\lambda$ for
      which all the $A$ vertices will be bought, but not $B$ vertices.
      More precisely, when $A$ vertices are bought,
      all the high-potential vertices get connected to the root,
      hence they become inactive. Also, the
      low-potential vertices will become inactive shortly after buying $A$, but
      before tightening $B$ vertices (this is achieved by setting appropriate
      perturbations to $A$ vertices as mentioned before).
      Now for some slightly larger initial potential, the low-potential
      vertices will also buy $B$ vertices before loosing their potential.
      Observe now, that the pruning phase will now keep all the $B$ vertices,
      but prune all the $A$ vertices.
      \begin{lemma}There exist the initial potential $\lambda$ such that, the
      LMP algorithm of Section~\ref{ss:moat-growing} returns the $A$ solution,
      while for the infinitesimally larger potential $\lambda^+$ it returns
      the $B$ solution.
      \end{lemma}
      Analogous arguments as in the result of
      Mestre~\cite{Mestre08:lagrangian-relax-partial-cover} can be
      used to deduce the following.
      \begin{corollary}\label{cor:lower-bound-internal}
        For any $r>1$ there is an infinite family of graphs where the
        natural moat growing algorithm for
        NW-PC-ST~\cite{DBLP:conf/swat/ByrkaLM16} has a ratio $r$ but
        where \emph{any} feasible solution to the NW-$k$-MST problem using
        only the nodes returned by this algorithm has cost at least
        ~$\nicefrac{4}{3}\cdot r$ times that of an optimum solution.
      \end{corollary}

    \end{subsubsection}
    \subsubsection{Interpretation}
    In the edge-weighted case of $k$-MST, Garg
    \cite{DBLP:conf/stoc/Garg05} was able to carefully exploit the
    inner workings of the Goemans-Williamson
    algorithm~\cite{DBLP:journals/siamcomp/GoemansW95} for the
    Lagrangian relaxation to match its ratio
    of~$2$. Corollary~\ref{cor:lower-bound-internal} means that our
    approach is in a certain sense optimal and that we would need to
    deviate from this framework to improve on the loss of factor
    $\nicefrac{4}{3}$ in the tree-merging step. This could possibly be
    achieved by exploiting structural properties of the underlying
    graph class or using nodes outside the solution returned by the
    LMP algorithm.

    Even when we exploit planarity it seems to be non-trivial to beat
    factor $4$ along the lines of Garg
    \cite{DBLP:conf/focs/Garg96,DBLP:conf/stoc/Garg05}.
    The changes in the solutions
    by increasing initial potentials of vertices can be much larger
    than those in the edge-weighted variant.  In particular, one can
    observe situations of node-flips in which two potentially distant
    vertices exchange their presence in the solution.  Also, in
    contrast to edge-weighted variant, a single node can be adjacent
    to any number of moats and not only two. This in turn causes the
    large difference in two trees produced by the algorithm. In
    particular, the OLD vertices as described by
    Garg~\cite{DBLP:conf/focs/Garg96} can form any number of connected
    components which may be expensive to connect even when the graph
    is planar.

  \section{Conclusions and comments}
    The $4+\eps$ approximation factor was obtained for the NW-$k$-MST problem
    on planar graphs. In the process we used the Lagrangian Relaxation technique.
    Our work can be interpreted as a generalization of a work on
    partial cover~\cite{KonemannPS11:Partial-Covering}.
    The result by Mestre~\cite{Mestre08:lagrangian-relax-partial-cover} implies
    that our factor is essentially best possible using the underlying LMP
    algorithm for the NW-PC-ST as a black-box. It shows that one would have
    to exploit planarity in the merging process to beat factor 4.

    Our ultimate hope would be to match the factor of $3$ of the LMP
    algorithm. We think that the question if this is possible is very
    interesting and challenging.
  \section{Acknowledgements}
    We would like to thank Zachary Friggstad
    for initial discussions on the problem.
    The authors were supported by the NCN grant number 2015/18/E/ST6/00456.

\bibliographystyle{abbrv}
\bibliography{bibliography}

\begin{thebibliography}{10}

\bibitem{DBLP:journals/mp/AroraK06}
S.~Arora and G.~Karakostas.
\newblock A $(2 + \varepsilon)$-approximation algorithm for the $k$-{MST}
  problem.
\newblock {\em Math. Program.}, 107(3):491--504, 2006.

\bibitem{DBLP:conf/icalp/BateniHL13}
M.~Bateni, M.~Hajiaghayi, and V.~Liaghat.
\newblock Improved approximation algorithms for (budgeted) node-weighted
  {Steiner} problems.
\newblock In {\em Proc. 40th International Colloquium on Automata, Languages,
  and Programming (ICALP'13)}, pages 81--92, 2013.

\bibitem{BermanY12:nw-nd}
P.~Berman and G.~Yaroslavtsev.
\newblock Primal-dual approximation algorithms for node-weighted network design
  in planar graphs.
\newblock In {\em Proc. 15th International Workshop on Approximation,
  Randomization, and Combinatorial Optimization (APPROX'12)}, pages 50--60,
  2012.

\bibitem{DBLP:conf/swat/ByrkaLM16}
J.~Byrka, M.~Lewandowski, and C.~Moldenhauer.
\newblock Approximation algorithms for node-weighted prize-collecting {Steiner}
  tree problems on planar graphs.
\newblock In {\em Proc. 15th Scandinavian Symposium and Workshops on Algorithm
  Theory (SWAT'16)}, pages 2:1--2:14, 2016.

\bibitem{DBLP:conf/icalp/ChekuriEV12}
C.~Chekuri, A.~Ene, and A.~Vakilian.
\newblock Node-weighted network design in planar and minor-closed families of
  graphs.
\newblock In {\em Proc. 39th International Colloquium on Automata, Languages,
  and Programming (ICALP'12)}, pages 206--217, 2012.

\bibitem{DBLP:conf/approx/ChekuriEV12}
C.~Chekuri, A.~Ene, and A.~Vakilian.
\newblock Prize-collecting survivable network design in node-weighted graphs.
\newblock In {\em Proc. 15th International Workshop on Approximation,
  Randomization, and Combinatorial Optimization (APPROX'12)}, pages 98--109,
  2012.

\bibitem{DBLP:journals/mp/ChudakRW04}
F.~A. Chudak, T.~Roughgarden, and D.~P. Williamson.
\newblock Approximate $k$-{MSTs} and $k$-{Steiner} trees via the primal-dual
  method and lagrangean relaxation.
\newblock {\em Math. Program.}, 100(2):411--421, 2004.

\bibitem{DBLP:journals/talg/DemaineHK14}
E.~D. Demaine, M.~T. Hajiaghayi, and P.~N. Klein.
\newblock Node-weighted {Steiner} tree and group {Steiner} tree in planar
  graphs.
\newblock {\em {ACM} Trans. Algorithms}, 10(3):13:1--13:20, 2014.

\bibitem{DBLP:conf/focs/Garg96}
N.~Garg.
\newblock A 3-approximation for the minimum tree spanning k vertices.
\newblock In {\em Proc. 37th Annual Symposium on Foundations of Computer
  Science (FOCS'96)}, pages 302--309, 1996.

\bibitem{DBLP:conf/stoc/Garg05}
N.~Garg.
\newblock Saving an epsilon: a 2-approximation for the $k$-{MST} problem in
  graphs.
\newblock In {\em Proc. 37th Annual {ACM} Symposium on Theory of Computing
  (STOC'05)}, pages 396--402, 2005.

\bibitem{DBLP:journals/siamcomp/GoemansW95}
M.~X. Goemans and D.~P. Williamson.
\newblock A general approximation technique for constrained forest problems.
\newblock {\em {SIAM} J. Comput.}, 24(2):296--317, 1995.

\bibitem{DBLP:journals/jacm/JainV01}
K.~Jain and V.~V. Vazirani.
\newblock Approximation algorithms for metric facility location and
  \emph{k}-median problems using the primal-dual schema and lagrangian
  relaxation.
\newblock {\em J. {ACM}}, 48(2):274--296, 2001.

\bibitem{KonemannPS11:Partial-Covering}
J.~K{\"{o}}nemann, O.~Parekh, and D.~Segev.
\newblock A unified approach to approximating partial covering problems.
\newblock {\em Algorithmica}, 59(4):489--509, 2011.

\bibitem{DBLP:conf/focs/KonemannSS13}
J.~K{\"{o}}nemann, S.~S. Sadeghabad, and L.~Sanit{\`{a}}.
\newblock An {LMP} {$O(\log n)$}-approximation algorithm for node weighted
  prize collecting {Steiner} tree.
\newblock In {\em Proc. 54th Annual {IEEE} Symposium on Foundations of Computer
  Science (FOCS'13)}, pages 568--577, 2013.

\bibitem{Mestre08:lagrangian-relax-partial-cover}
J.~Mestre.
\newblock Lagrangian relaxation and partial cover.
\newblock In {\em Proc. 25th Annual Symposium on Theoretical Aspects of
  Computer Science (STACS'08)}, pages 539--550, 2008.

\bibitem{DBLP:journals/iandc/Moldenhauer13}
C.~Moldenhauer.
\newblock Primal-dual approximation algorithms for node-weighted {Steiner}
  forest on planar graphs.
\newblock {\em Inf. Comput.}, 222:293--306, 2013.

\bibitem{DBLP:journals/siamcomp/MossR07}
A.~Moss and Y.~Rabani.
\newblock Approximation algorithms for constrained node weighted {Steiner} tree
  problems.
\newblock {\em {SIAM} J. Comput.}, 37(2):460--481, 2007.

\bibitem{SinaMasterThesis}
{Sadeghian Sadeghabad, Sina}.
\newblock Node-weighted prize-collecting {Steiner} tree and applications.
\newblock Master's thesis, 2013.

\end{thebibliography}

\appendix

\section{Proof of Lemma~\ref{lem:four+eps-appr-algor}}
\label{sec:proof-lemma-refl}
  To bound the cost of the cheaper of two solutions $\mathrm{SOL}_1$ and
  $\mathrm{SOL}_2$ we employ the following Lemma by Könemann et
  al. \cite{KonemannPS11:Partial-Covering}.
  \begin{lemma}[\cite{KonemannPS11:Partial-Covering}]\label{lem:alpha-beta-r-estimation} For any $r>1$ and $\delta>0$, we have
    \begin{displaymath}
      \max_{\substack{\alpha\in(0,1)\\ \beta\in[0,r]}}\min\left\{\frac{r(1+\delta)-(1-\alpha)\beta}{\alpha},r(1+\delta)+\alpha\beta\right\}=\left(\frac43+O(\sqrt{\delta})\right)r\,.
     \end{displaymath}\qed{}
   \end{lemma}
  \begin{proof}[Proof of Lemma~\ref{lem:four+eps-appr-algor}]
    Let $\alpha = \alpha_2$ and $\beta=\frac{c(T_1)}{OPT}$. With this notation we obtain in a similar way as Könemann et
    al. \cite{KonemannPS11:Partial-Covering}
    \begin{align*}
      c(\mathrm{SOL}_1) &\leq c(T_1) + (1+\eps_2)\alpha\cdot c(T_2)
                          + \eps \cdot \oh(\log(1/\eps_2)) \cdot \opt \\
        & \leq  \alpha\cdot c(T_1) + (1-\alpha)\cdot c(T_1)+(1+\eps_2)\alpha\cdot c(T_2)+ \eps \cdot \oh(\log(1/\eps_2)) \cdot \opt\\
        & \leq \left(3(1+\eps_2) + \alpha \beta \right) \cdot \opt
        + \eps \cdot \oh(\log(1/\eps_2)) \cdot \opt, \\
    \end{align*}
    and
    \begin{align*}
            c(\mathrm{SOL}_2) & = c(T_2) \\
                              & = \frac{\alpha\cdot c(T_2)}{\alpha}\\
                              & \leq \frac{(3+\eps)\opt-(1-\alpha)c(T_1)}{\alpha}\\
                              & \leq \frac{3(1+\eps)-(1-\alpha)\beta}{\alpha}
                                \cdot \opt\,.
    \end{align*}
    By setting $r=3$ and $\delta=\eps=\eps_2$ we obtain via
    Lemma~\ref{lem:alpha-beta-r-estimation} that the better of the two
    solutions has cost no more than
    $(4+O(\sqrt{\eps}+\eps\log\nicefrac{1}{\eps}))\cdot
    \opt=(4+O(\sqrt{\eps}))\cdot
    \opt$ completing the proof.\qed
  \end{proof}

\end{document}